\documentclass[%
 reprint,
 amsmath,amssymb,
 aps,
]{revtex4-2}

\usepackage{graphicx}
\usepackage{dcolumn}
\usepackage{bm}

\usepackage{physics}
\usepackage{amsthm}
\usepackage{mathtools}
\usepackage{verbatim}

\DeclarePairedDelimiter{\ceil}{\lceil}{\rceil}

\DeclarePairedDelimiter{\floor}{\lfloor}{\rfloor}

\theoremstyle{theorem}
\newtheorem{theorem}{Theorem}
\newtheorem{prop}{Proposition}
\newtheorem{lemma}{Lemma}
\newtheorem{corollary}{Corollary}[theorem]

\theoremstyle{remark}

\theoremstyle{plain}

\theoremstyle{definition}

\newcommand{\hil}{\mathcal H}

\newcommand{\C}{\mathbb C}
\newcommand{\ts}[2]{\left(#1\right)^{\otimes#2}}

\begin{document}

\preprint{APS/123-QED}

\title{On generating $r$-uniform subspaces with the isometric mapping method}

\author{K. V. Antipin}
 \email{kv.antipin@physics.msu.ru}
 \affiliation{Faculty of Physics, M. V. Lomonosov Moscow State University,\\ Leninskie gory, Moscow 119991, Russia}

\date{\today}

\begin{abstract}
We propose a compositional approach  to construct subspaces consisting entirely  of $r$-uniform states, including the ones in heterogeneous systems. The approach allows one to construct new objects from old ones: it combines encoding isometries of pure quantum error correcting codes with entangled multipartite states and  subspaces. The presented methods can be also used to construct new pure quantum error correcting codes from certain combinations of  old ones. The approach is illustrated with various examples including constructions of $2$-, $3$-,  $4$-, $5$-uniform subspaces. The results are then compared with analogous constructions obtained with the use of orthogonal arrays. 
\end{abstract}

\maketitle


\section{\label{sec:intro}Introduction}

Multipartite entanglement is crucial for realization of various protocols of quantum information processing~\cite{JzL03,RB01,Scott04,MEO18}. One important manifestation of this phenomenon is genuine multipartite entanglement~(GME)~\cite{Svet87,GHZ92,DVC00}. In GME states  entanglement is present in every bipartite cut of a compound system, which makes them useful in  communication protocols such as quantum teleportation and dense coding~\cite{YeCh06,MP08}. Another interesting form  is $r$-uniform~(also known as maximal) entanglement~\cite{Scott04,FFPP08,AC13}. Each reduction of an $r$-uniform state to $r$ subsystems is maximally mixed. This property is closely related to quantum secret sharing~\cite{CGL99,HCLRL12} and quantum error correcting codes~(QECCs)~\cite{KLaf97,HuGra20}.

Recently the notion of entangled subspaces  has been attracting  much attention owing to its potential use in quantum information science. It was first  described in Ref.~\cite{Parth04}, where the term ``completely entangled subspaces'' was coined. Later, depending on the form of multipartite entanglement present in each state of a subspace, several other types were introduced: genuinely entangled subspaces~(GESs)~\cite{DemAugWit18}, negative partial transpose~(NPT) subspaces~\cite{NJohn13},  $r$-uniform subspaces~(rUSs)~\cite{HuGra20}. In the present paper we concentrate on construction of $r$-uniform subspaces, mostly for heterogeneous systems, i.~e., those having different local dimensions.  There are a number of tools for constructing $r$-uniform states in homogeneous systems:  graph states~\cite{Helw13}, elements of combinatorial design such as Latin squares~\cite{GALRZ15}, symmetric matrices~\cite{FJXY15}, orthogonal arrays~(OAs)~\cite{GZ14} and their  variations~\cite{PZLZ19,PZDW20,PXC22}. For construction of $r$-uniform states in heterogeneous systems OAs  were extended to mixed orthogonal arrays~(MOAs)~\cite{GBZ16}. Recent developments of this method can be found in Refs.~\cite{PZFZ21,SSCZ22}. The main source for $r$-uniform subspaces in homogeneous systems are pure quantum error correcting codes~\cite{Scott04, HuGra20}. Little is known about construction  of $r$-uniform subspaces in heterogeneous systems~(the only method we could find in literature was based on Proposition 12 of Ref.~\cite{SSCZ22}). Development of new methods of construction of rUSs for this case is our main motivation for the present paper. $R$-uniform subspaces in heterogeneous systems have relation  to QECCs over mixed alphabets~\cite{WYFO13} and quantum information masking~\cite{SLCZ21}. To our knowledge, for a given system the largest possible dimension of rUSs is unknown, so building new instances of such subspaces could bring some insights in this question.

We use compositional tools of diagrammatic reasoning~\cite{CoeKis17,Biamonte19,BiamonteEtAll15}, which allow us to come up with new constructions and provide further instances of states and subspaces with important properties. Tensor diagrams are widely used in quantum information  theory, in particular, in theory of QECCs. Recently a framework for the construction of new stabilizer QECCs from old ones with the use of tensor networks has been presented in Ref.~\cite{CaoLack22}.

The paper is organized as follows. In Section~\ref{sec:prel} necessary definitions  and some theoretical background are given. The main results of the current paper are provided in Section~\ref{sec:res}. In Subsection~\ref{sec:res:runi} we give diagrammatic representation of  basic properties of rUSs upon which, in Subsection~\ref{sec:res:build}, we derive the methods of constructing rUSs in heterogeneous systems such as glueing several subspaces together, eliminating  parties, combining pure error correcting codes and maximally entangled states and subspaces. In Subsection~\ref{sec:res:comp} we compare our results with the ones obtained with the use of the mixed orthogonal arrays method. Finally, in Section~\ref{sec:conc} we conclude with discussing possible directions of further research.

\section{\label{sec:prel} Preliminaries}

Let us first give the definition of $r$-uniform states of an $n$-partite finite-dimensional system with local dimensions $d_1,\,\ldots,\,d_n$. Such a system is usually associated with the tensor product Hilbert space $\mathbb{C}^{d_1}\otimes\ldots\otimes\mathbb{C}^{d_n}$.
A state $\ket{\psi}$ in $\mathbb{C}^{d_1}\otimes\ldots\otimes\mathbb{C}^{d_n}$ is called \emph{$r$-uniform} if all its reductions \emph{at least} to $r$ parties are maximally mixed, i.~e., 
\begin{equation}
    \mathrm{Tr}_{\{i_1,\,\ldots,\,i_r\}^c}[\dyad{\psi}]=\frac1{ d_{i_1}\cdot\ldots\cdot d_{i_r}}\,I_{i_1,\,\ldots,\,i_r}
\end{equation}
for all $r$-element subsets $\{i_1,\,\ldots,\,i_r\}$ of the set $\{1,\, \ldots,\, n\}$. Here $\{i_1,\,\ldots,\,i_r\}^c$  denotes the complement of the given set in  the set of all parties. It is clear that $r$-uniform state is also $l$-uniform for all $l<r$. By the properties of the Schmidt decomposition, the necessary condition for $r$-uniform states to exist is that $d_{i_1}\cdot\ldots\cdot d_{i_r}\leqslant d_{i_{r+1}}\cdot\ldots\cdot d_{i_n}$ is satisfied for each bipartition $i_1,\,\ldots,\,i_r |i_{r+1},\,\ldots,\,i_n$. 

An \emph{$r$-uniform subspace} --- a subspace of $\mathbb{C}^{d_1}\otimes\ldots\otimes\mathbb{C}^{d_n}$ consisting entirely of $r$-uniform vectors.

\begin{figure}[t]
\includegraphics[scale=0.3]{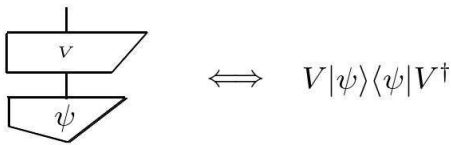}
\caption{\label{fig:dblV} Doubling notation for the process of action of a linear operator $V$ on a pure state $\psi$}
\end{figure}

\begin{figure}[t]
\includegraphics[scale=0.3]{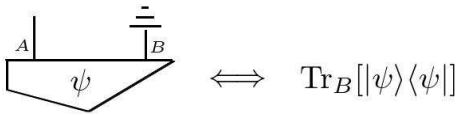}
\caption{\label{fig:dblT} Reduction of a bipartite pure state $\psi$ to subsystem $A$}
\end{figure}

For homogeneous systems, i.~e., those having equal local dimensions, the existence of $r$-uniform subspaces can be deduced from the existence of certain quantum error correcting codes~(QECCs). Recall that a QECC $((n,\,K,\,d))_D$ is a special  $K$-dimensional subspace of $\ts{\C^D}{n}$ such that for each its state any error affecting not more than a certain number of subsystems can be corrected. For a code with distance $d=2t+1$ the number is equal to $t$. In addition, a code with distance $d$ can detect $d-1$ errors.

In addition to the  $((n,\,K,\,d))_D$ notation for general QECCs, we will use the designation $[[n,\,k,\,d]]_D$ for stabilizer QECCs. While the symbols $n$ and $d$ from the latter notation have the same sense as those in the former one, the dimension of the codespace for the code $[[n,\,k,\,d]]_D$  is equal to $D^k$.

A quantum error correcting code  is called \emph{pure} if 
\begin{equation}
\bra{i}E\ket{j}=0
\end{equation}
for any states $\ket{i},\,\ket{j}$ from an orthonormal set spanning the code space and for any error operator $E$ with weight strictly less than the distance of the code.

It is known that each pure QECC $((n, K, d))_D$ yields a $K$-dimensional $(d-1)$-uniform subspace of $(\mathbb C^D)^{\otimes n}$, and vice versa~\cite{HuGra20}.

To address the case of heterogeneous systems, in the present paper we will use encoding isometries of the existing pure QECCs in combination with various states and subspaces of lower number of parties. A similar approach dealing with isometric mapping to entangled subspaces proved to be effective in constructing multipartite genuinely entangled subspaces~\cite{KVAnt21}. 

Throughout the paper we use tensor diagrams, in particular, we use doubled-process theory notation adopted from Ref.~\cite{CoeKis17}. The doubling notation indicates the passage from pure state vectors to  their associated density operators, as shown on Fig.~\ref{fig:dblV}. 

To deal also with mixed states, the \emph{discarding symbol~(map)} is used. Applying the discarding map to a subsystem of  a multipartite state is equivalent to tracing out the subsystem, as shown on Fig.~\ref{fig:dblT}.

\begin{figure}[t]
\includegraphics[scale=0.3]{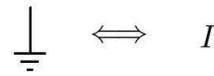}
\caption{\label{fig:mxms} Diagrammatic representation of the maximally mixed state~(up to the normalization factor).}
\end{figure}

The adjoint of the discarding map~(see Fig.~\ref{fig:mxms}) denotes the identity operator, which is proportional to the maximally mixed state.

\section{\label{sec:res}Results}

\subsection{\label{sec:res:runi}Basic properties and their diagrammatic representation}

We start with pointing at an important basic property of  subspaces under consideration. Let $\ket{\phi}$ and $\ket{\chi}$ be two mutually orthogonal normalized vectors in an $r$-uniform subspace $W$. An arbitrary~(normalized) linear combination $\ket{\psi}=\alpha\ket{\phi} + \beta\ket{\chi}$ is also in $W$, and hence its reduction to some $r$-element subset $S$ of the set of all parties yields
\begin{multline}
    \mathrm{Tr}_{S^c}[\dyad{\psi}]=  N\,I_S\\
    = N\,I_S + \alpha\beta^*\,\mathrm{Tr}_{S^c}[\dyad{\phi}{\chi}] + \beta\alpha^*\,\mathrm{Tr}_{S^c}[\dyad{\chi}{\phi}],
\end{multline}
where $N=\left(\prod_{i\in S}d_i\right)^{-1}$, the normalization factor. Consequently, the last two terms sum up to zero:
\begin{multline}
    \alpha\beta^*\,\mathrm{Tr}_{S^c}[\dyad{\phi}{\chi}] + \beta\alpha^*\,\mathrm{Tr}_{S^c}[\dyad{\chi}{\phi}] = 0,\\
    \forall\,\alpha,\,\beta\in\C, \quad\abs{\alpha}^2+\abs{\beta}^2 = 1.
\end{multline}
Setting first $\alpha$ real and $\beta$ imaginary and then both of them  real, one can deduce that
\begin{equation}
    \mathrm{Tr}_{S^c}[\dyad{\phi}{\chi}] = 0.
\end{equation}
Now we can formulate this property as
\begin{lemma}\label{rlem}
   For any orthonormal set $\{\ket{\psi}_i\}$ spanning $r$-uniform subspace it follows that
\begin{equation}\label{TrOrth}
    \mathrm{Tr}_{S^c}[\dyad{\psi_i}{\psi_j}] \sim \delta_{ij}\,I_S
\end{equation}
for any $r$-element subset $S$ of the set of all parties. 
\end{lemma}

Eq.~(\ref{TrOrth}) is known to hold for codewords of pure QECCs of distance $r+1$, which correspond to $r$-uniform subspaces for homogeneous systems, but it is worth stressing that it is valid for more general case of heterogeneous systems.

\begin{figure}[t]
\includegraphics[scale=0.3]{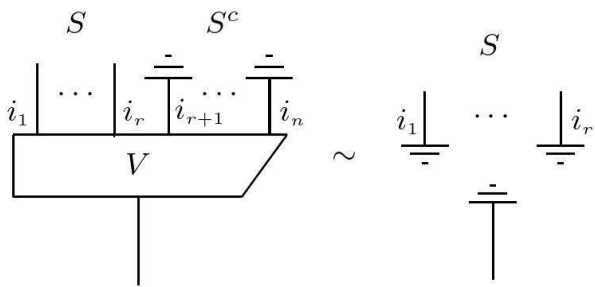}
\caption{\label{fig:rIso} Action of $V$ together with tracing out subsystems $S^c$ results in a channel $\Phi_S$ which discards its input and returns the maximally mixed state.}
\end{figure}

Let us think of $r$-uniform subspaces in terms of isometries and quantum channels. To each such subspace $W$ with dimension $K$ one can relate an isometry $V\colon\mathbb{C}^K\rightarrow \mathbb{C}^{d_1}\otimes\ldots\otimes\mathbb{C}^{d_n}$, which maps an orthonormal basis $\{\ket{i}\}$ of $\mathbb{C}^K$ to some orthonormal set $\{\ket{\psi}_i\}$ spanning $W$:
\begin{equation}\label{Isoi}
    V\ket{i} = \ket{\psi_i}.
\end{equation}
Hence, the range of the isometry $V$ coincides with the subspace $W$. As before, let us choose an $r$-element subset $S$ of the set of $n$ parties. Applying the isometry to a state in $\mathbb{C}^K$  with subsequent tracing out $n-r$ subsystems in $S^c$ results in action of a quantum channel on the state:  
\begin{equation}\label{ChanV}
    \mathrm{Tr}_{S^c}[V\dyad{\phi}V^{\dagger}] = \Phi_S(\dyad{\phi}),\quad\ket{\phi}\in\mathbb{C}^K.
\end{equation}
We thus obtain a family of quantum channels $\Phi_S\colon L(\mathbb{C}^K)\rightarrow L(\mathbb{C}^{d_{i_1}}\otimes\ldots\otimes\mathbb{C}^{d_{i_r}})$, where $\{i_1,\,\ldots,\,i_r\}=S$; one channel for each choice of $S$. Since subspace $W$ is $r$-uniform, a channel $\Phi_S$ maps all states of $\mathbb{C}^K$  to the identity on $S$:
\begin{equation}
    \Phi_S(\dyad{\phi}) = \frac1{ d_{i_1}\cdot\ldots\cdot d_{i_r}}\,I_{i_1,\,\ldots,\,i_r},\quad\ket{\phi}\in\mathbb{C}^K.
\end{equation}
In other words, the channels  discard the input and map everything to the maximally mixed state:
\begin{equation}\label{TrChan}
    \Phi_S(X) = \frac{\mathrm{Tr}[X]}{ d_{i_1}\cdot\ldots\cdot d_{i_r}}\,I_S,\quad X\in L(\mathbb{C}^K).
\end{equation}
By setting $X=\dyad{i}{j}$ in this expression, with the use of Eqs.~(\ref{Isoi}), (\ref{ChanV}), we recover Eq.~(\ref{TrOrth}).

It should be stressed that the above property of mapping to the maximally mixed state holds when the input dimension of the isometry~(and of the corresponding channel) is not greater than the dimension of the range of the isometry, i.~e., the dimension of the $r$-uniform subspace. In fact, the input dimension can be strictly less than that dimension, in which case the isometry takes the input states to some subspace of the $r$-uniform space.

Diagrammatic representation of Eqs.~(\ref{ChanV}) and (\ref{TrChan}) is shown on Fig.~\ref{fig:rIso}, where the symbol "$\sim$" means that the two diagrams are equal up to a scalar factor~(the normalization constant for the maximally mixed state). This construction will be crucial in further considerations.

Now we can choose $V$ to be an encoding isometry of some pure quantum error correcting code. Let us try to combine this isometry with some states.

\subsubsection*{Example: 2-uniform state in heterogeneous systems}

Consider the $((5,\, 3,\, 3))_3$ pure code~\cite{Rains99} and its encoding isometry $V$. Application of $V$ to one of the parties of a bipartite pure state $\ket{\psi}$ in $\mathbb C^2\otimes\mathbb C^3$ yields a $6$-partite pure state in $\C^2\otimes\ts{\C^3}{5}$, as shown on Fig.~\ref{fig:23st2}.  The code has distance $3$, so the code subspace is $2$-uniform. In addition, the code subspace has dimension equal to 3, which matches the local dimension of the second party of the state $\ket{\psi}$. Therefore, the property of Fig.~\ref{fig:rIso} holds in this case  with $r$ = 2.

\begin{figure}[t]
\includegraphics[scale=0.3]{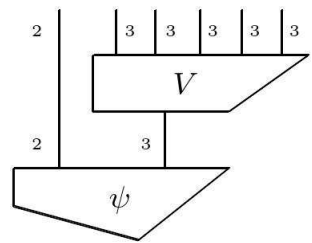}
\caption{\label{fig:23st2} Construction of a $2$-uniform state from an encoding isometry $V$ of the $((5,\, 3,\, 3))_3$ pure code and a maximally entangled state $\ket{\psi}$ in $\C^2\otimes\C^3$.}
\end{figure}

If the bipartite state $\ket{\psi}$ is \emph{maximally entangled}, i. e. its reduction to the party with local dimension $2$ is maximally mixed, then the resulting state in $\C^2\otimes\ts{\C^3}{5}$ will be $2$-uniform. This can be shown diagrammatically. One needs to consider the two cases of producing the two-party reduction of the state in question: a) all parties are traced out except some two output subsystems of $V$; b) all parties are traced out except the first party of $\ket{\psi}$~(with dimension $2$) and some output subsystem of $V$.

\begin{figure}[b]
\includegraphics[scale=0.3]{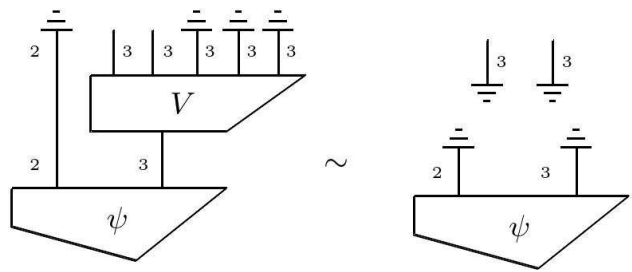}
\caption{\label{fig:2unitra} The state $\ket{\psi}$ is completely traced out -  the part of the diagram on the bottom right  is a scalar equal to $1$.}
\end{figure}

The case a) is presented on Fig.~\ref{fig:2unitra}: the property on Fig.~\ref{fig:rIso} being used, the state $\ket{\psi}$ gets completely traced out and the resulting state is proportional to $I_3\otimes I_3$.

The case b) is analyzed on Fig.~\ref{fig:2unitrb}: on the first step the property on Fig.~\ref{fig:rIso} is used; the second step is due to the fact that  $\ket{\psi}$ is maximally entangled.

It is interesting to note that $\C^2\otimes\ts{\C^3}{5}$ was the smallest possible Hilbert space for which a $2$-uniform state could be constructed with the methods of Ref.~\cite{GBZ16}.

\subsection{\label{sec:res:build} Construction of $r$-uniform subspaces in heterogeneous systems}

The simplest method to produce an $r$-uniform subspace in heterogeneous systems is to "glue" together two $r$-uniform subspaces in homogeneous systems. By "glueing" we mean taking tensor product of the two subspaces: this can be done by taking  all possible  tensor products of the vectors spanning the two subspaces, the resulting subspace will be spanned by such combinations.

\begin{figure}[t]
\includegraphics[scale=0.3]{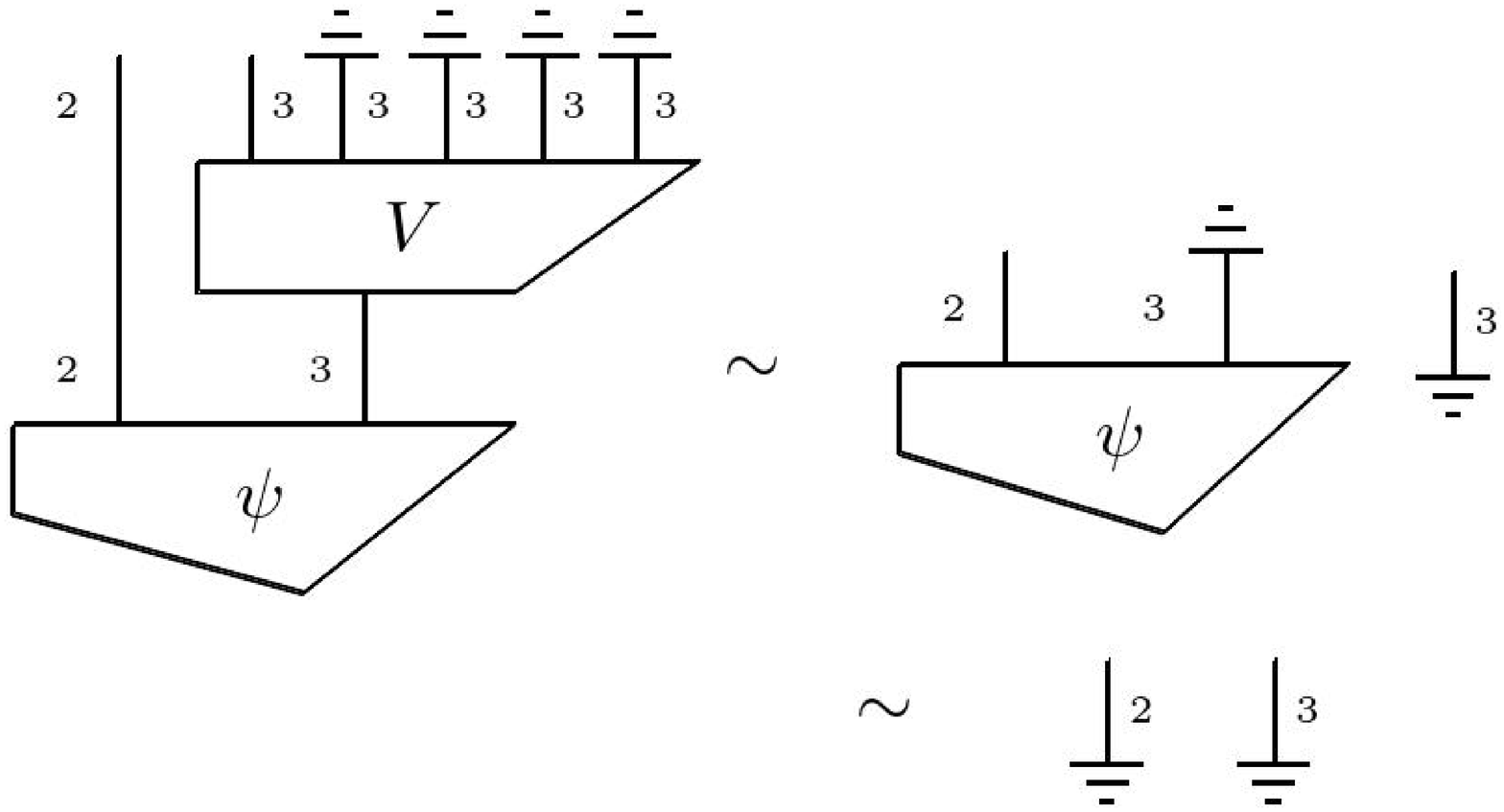}
\caption{\label{fig:2unitrb} By the property on Fig.~\ref{fig:rIso} and   the maximal entanglement of $\ket{\psi}$ the resulting state is proportional to $I_2\otimes I_3$.}
\end{figure}

\begin{lemma}\label{glue}
    Tensor product of an $r$-uniform subspace and a $k$-uniform subspace is an $l$-uniform subspace, where $l=\min(r,\,k)$.
\end{lemma}

\begin{proof}
    Let $W_1$ be an $r$-uniform subspace with $n$ parties  and let $W_2$ be a $k$-uniform subspace with $m$ parties. Consider two isometries $V_1\colon\,\hil_{A}\rightarrow\hil_{C_1}\otimes\cdots\otimes\hil_{C_n}$ and $V_2\colon\,\hil_{B}\rightarrow\hil_{D_1}\otimes\cdots\otimes\hil_{D_m}$, where $\dim(\hil_A)=\dim(W_1)$ and $\dim(\hil_B)=\dim(W_2)$. The first one, $V_1$, maps the basis states $\{\ket{i}\}_A$ of $\hil_A$ to an orthonormal system of vectors spanning $W_1$. Similarly, $V_2$ maps the basis states $\{\ket{j}\}_B$ of $\hil_B$ to vectors spanning $W_2$. Tensor product $W_1\otimes W_2$ is then spanned by the vectors 
    \begin{equation}
        \left(V_1\otimes V_2\right)\left(\ket{i}_{ A}\otimes\ket{j}_{B}\right),
    \end{equation}
    as shown on Fig.~\ref{fig:2isouni}. Now the property from Fig.~\ref{fig:rIso} can be applied when one traces out \emph{any} $n+m-l$ of the parties $C_1,\,\ldots,\,C_n,\,D_1,\,\ldots,\,D_m$. As a result, a general state $\ket{\phi}\in\hil_A\otimes\hil_B$ gets completely traced out, and the $l$-party maximally mixed state is produced.
\end{proof}

$R$-uniform subspaces can be used for $r$-uniform quantum information masking~\cite{SLCZ21}. An operation $V$ is said to $r$-uniformly mask quantum information contained in states $\{\ket{i}\}$ if it maps them to multipartite states $\{\ket{\psi_i}\}$ whose all reductions to $r$ parties are identical. In the proof of Lemma~\ref{glue}  an instance of masking has been provided: on the right part of Fig.~\ref{fig:2isouni} it is shown how each state  $\phi$ from $\hil_A\otimes\hil_B$ is ``masked'' by the two isometries $V_1$ and $V_2$ as an $l$-uniform state.

As an example, combining encoding isometries of $((5,\,2,\,3))_2$ and $((5,\,3,\,3))_3$ pure codes, by Lemma~\ref{glue} we obtain a $6$-dimensional $2$-uniform subspace of the $\ts{\C^2}{5}\otimes\ts{\C^3}{5}$ Hilbert space.

Can we reduce the number of parties? A structure similar to the one on Fig.~\ref{fig:23st2} can be used. Let us take the encoding isometry $V$ of the $[[6,\,2,\,3]]_3$ (stabilizer) pure code~(Ref.~\cite{JLLX10}, Corollary 3.6). The range of the isometry is a  $2$-uniform subspace of the $\ts{\C^3}{6}$ Hilbert space, which has dimension equal to $3^2=9$. Consider a subspace of $\C^2\otimes\C^9$, which consists entirely of states maximally entangled with respect to the first party. Such a subspace can be easily constructed with the use  of Proposition 3 of Ref.~\cite{GW07}. The dimension of the subspace is equal to $\floor{\frac92}=4$~(Corollary 4 of Ref.~\cite{GW07}). Now let us act with $V$ on the second party~(the one with dimension $9$) of each state in the subspace. This procedure will generate a $4$-dimensional subspace of the $\C^2\otimes\ts{\C^3}{6}$ Hilbert space. The analysis, which is similar to that on Figs.~\ref{fig:2unitra} and~\ref{fig:2unitrb}, shows that the resulting subspace is $2$-uniform.

\begin{figure}[t]
\includegraphics[scale=0.35]{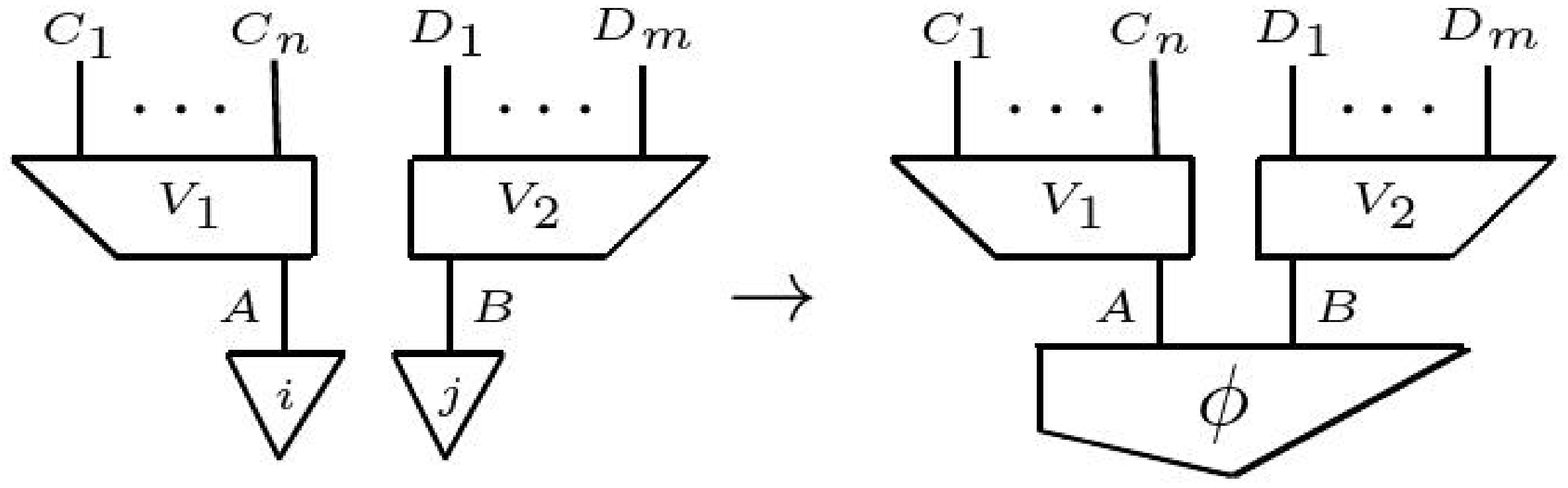}
\caption{\label{fig:2isouni} On the left: action of the isometry $V_1\otimes V_2$ on a particular basis state $\ket{i}_{A}\otimes\ket{j}_{B}$ of $\hil_A\otimes\hil_B$. On the right: action of the isometry $V_1\otimes V_2$ on a general  state $\phi$ from $\hil_A\otimes\hil_B$, which is equal to a linear combination  of basis states $\{\ket{i}_{A}\otimes\ket{j}_{B}\}$. Action of   $V_1\otimes V_2$ on each state in $\hil_A\otimes\hil_B$ generates $W_1\otimes W_2$, which is $l$-uniform.}
\end{figure}

The above construction can be viewed as an illustration of \emph{subspace masking}: only states that belong to a specific subspace of $\C^2\otimes\C^9$ are masked by $V$ as $2$-uniform states in $\C^2\otimes\ts{\C^3}{6}$.

The next property provides an important way of generating $r$-uniform subspaces from those of larger number of parties. It can be seen as the extension of Theorem~20 of Ref.~\cite{Rains98} to the case of heterogeneous systems.
\begin{theorem}\label{TrGen}
    Let $W$ be an $r$-uniform subspace of Hilbert space with the set $S=\{d_1,\,\ldots,\,d_n\}$ of local dimensions. Let $r\geqslant1$ and $\dim(W)=K$. Then, for any $d_i\in S$,  there exists an $r-1$-uniform subspace of Hilbert space with local dimensions $S\setminus \{d_i\}$. The dimension of this subspace is equal to $d_i K$.
\end{theorem}
\begin{proof}
    Consider an orthonormal set of vectors $\{\ket{\psi_k}\}$ which span $W$. Each vector $\ket{\psi_k}$ is  $r$-uniform, and so, in particular, its reduction to  party $i$ is proportional to the maximally mixed operator $I_{\{i\}}$. Accordingly, the Schmidt decomposition of $\ket{\psi_k}$ with respect to the bipartition ''party i$|$the rest'' reads
    \begin{equation}\label{SchU}
        \ket{\psi_k} = \frac1{\sqrt{d_i}}\sum_{j=0}^{d_i-1}\,\ket{\phi^{(k)}_j}_{\overline{P}}\otimes\ket{\chi^{(k)}_j}_P,
    \end{equation}
    where $\{\ket{\phi^{(k)}_j}_{\overline{P}}\}$ and $\{\ket{\chi^{(k)}_j}_P\}$, with $j=0,\,\ldots,\,d_i-1$ and $k$ fixed, are two orthonormal sets of vectors in $r-1$-partite and $1$-partite Hilbert spaces with local parties $\overline{P}=\{1,\,\ldots,\,n\}\setminus\{i\}$ and $P=\{i\}$, respectively. In the right part of Eq.~(\ref{SchU}) vectors with the \emph{same} upper index satisfy the orthonormality condition, for example,
    \begin{equation}\label{OrthC}       \bra{\chi^{(k)}_m}\ket{\chi^{(k)}_n}_P=\delta_{mn}.
    \end{equation}
    
    Now consider an $r-1$-element subset  $J\subset\overline{P}$ and, for some numbers $k,\,s\in\{1,\,2,\,\ldots,\,K\}$, take the reduction of $\dyad{\psi_k}{\psi_s}$ to the set $J\cup\{i\}$:
    \begin{multline}
        \mathrm{Tr}_{\overline{P}\setminus J}\left[\dyad{\psi_k}{\psi_s}\right]\\ = \frac1{d_i}\sum_{j,\,l=0}^{d_i-1}\,\mathrm{Tr}_{\overline{P}\setminus J}\left[\dyad{\phi^{(k)}_j}{\phi^{(s)}_l}_{\overline{P}}\right]
        \otimes\dyad{\chi^{(k)}_j}{\chi^{(s)}_l}_P.
     \end{multline}
By Lemma~\ref{rlem}, this reduction is proportional to $\delta_{ks}\,I_{J\cup\{i\}}=\delta_{ks}\,I_J\otimes I_{\{i\}}$, and hence
\begin{multline}
    \frac1{d_i}\sum_{j,\,l=0}^{d_i-1}\,\mathrm{Tr}_{\overline{P}\setminus J}\left[\dyad{\phi^{(k)}_j}{\phi^{(s)}_l}_{\overline{P}}\right]
        \otimes\dyad{\chi^{(k)}_j}{\chi^{(s)}_l}_P \\
        =\delta_{ks}\,\frac1{d_J}\,I_J\otimes\,\frac1{d_i}\,I_{\{i\}},
\end{multline}
where $d_J$ is the product of local dimensions of the parties in $J$. Multiplying both parts of this equality by $\bra{\chi^{(k)}_m}$ and $\ket{\chi^{(s)}_n}$, with the use of condition~(\ref{OrthC}), we obtain
\begin{equation}\label{TrSub}
    \mathrm{Tr}_{\overline{P}\setminus J}\left[\dyad{\phi^{(k)}_m}{\phi^{(s)}_n}_{\overline{P}}\right] = \delta_{ks}\delta_{mn}\,\frac1{d_J}\,I_J.
\end{equation}
Taking trace over subsystem $J$ in the last equation, one can see that the set of $d_iK$ vectors $\{\ket{\phi^{(t)}_j}_{\overline{P}}\}$,  $j=0,\,\ldots,\,d_i-1$, $t=1,\,\ldots,\,K$ is an orthonormal system.
Since Eq.~(\ref{TrSub}) holds for any choice of an $r-1$-element subset $J\subset\overline{P}$, the states $\{\ket{\phi^{(t)}_j}_{\overline{P}}\}$ are $r-1$-uniform. In addition, from the same equation one can see that any linear combination of these vectors  is an $r-1$-uniform state. Therefore, the system  $\{\ket{\phi^{(t)}_j}_{\overline{P}}\}$ spans a $d_iK$-dimensional $r-1$-uniform subspace. From Eq.~(\ref{SchU}) it follows that 
\begin{equation}\label{ProjSub}
    d_i\mathrm{Tr}_{\{i\}}\left[\sum_{s=1}^K\,\dyad{\psi_s}\right]
\end{equation}
is   the orthogonal projector on the subspace in question.
\end{proof}

A practical way to obtain an orthonormal system of vectors spanning the subspace defined by the projector in Eq.~(\ref{ProjSub}) is as follows. Let us suppose that party $i$ with local dimension $d_i$ is being eliminated, just as in the condition of Theorem \ref{TrGen}. Consider an orthonormal system of one party vectors $\{\ket{v_j}\},\,j=0,\,\ldots,\,d_i-1$, which spans $\C^{d_i}$ and such that each partial scalar product  
\begin{multline}   \ket{\mu_j^{(s)}}\equiv\bra{v_j}\ket{\psi_s},\qquad j=0,\,\ldots,\,d_i-1,\\s=1,\,\ldots,\,K,
\end{multline}
is a non-null vector, where the only input of the  (co)vector $\bra{v_j}$ is joined with the $i$-th output of the vector $\ket{\psi_s}$ in each partial scalar product. Then $d_iK$ vectors $\{\ket{\mu_j^{(s)}}\}$ span the subspace in question. Indeed, the vectors are mutually orthogonal:
\begin{multline}
    \bra{\mu_l^{(t)}}\ket{\mu_j^{(s)}}=\bra{\psi_t}\ket{v_l}\bra{v_j}\ket{\psi_s}\\
    =\bra{v_j}\mathrm{Tr}_{\overline{P}}\left[\dyad{\psi_s}{\psi_t}\right]\ket{v_l}=\frac1{d_i}\delta_{ts}\bra{v_j}\ket{v_l}\\
    =\frac1{d_i}\delta_{ts}\delta_{lj},
\end{multline}
where the third equality follows from $r$-uniformity of the original vectors $\{\ket{\psi_s}\}$ and Lemma~\ref{rlem}. Consequently, $\{\sqrt{d_i}\ket{\mu_j^{(s)}}\}$ is an orthonormal system, and the corresponding projector 
\begin{multline}
    \sum_{j,\,s}\,\sqrt{d_i}\dyad{\mu^{(s)}_j}\sqrt{d_i}=d_i\sum_{j,\,s}\bra{v_j}\ket{\psi_s}\bra{\psi_s}\ket{v_j}\\
    =d_i\mathrm{Tr}_{\{i\}}\left[\sum_{s}\,\dyad{\psi_s}\right]
\end{multline}
coincides with the one in Eq.~(\ref{ProjSub}).

As an example, from the obtained above $4$-dimensional $2$-uniform subspace of $\C^2\otimes\ts{\C^3}{6}$ Hilbert space one can produce a $12$-dimensional $1$-uniform subspace of $\C^2\otimes\ts{\C^3}{5}$ Hilbert space  by eliminating one of the parties with local dimension $3$.

When an initial $r$-uniform subspace is spanned just by $1$ vector, the described above practical method becomes similar to Proposition 12 of Ref.~\cite{SSCZ22}.

As we have seen from Lemma~\ref{glue}, glueing two uniform subspaces together doesn't increase the uniformity parameter of the resulting subspace. This is in accordance with the general principle that local operations cannot produce any entanglement over that present in original states. Let us show that making use of additional resources such as maximally entangled states can  lead to larger uniformity parameters of the produced states and subspaces in comparison with original ones. At first we  consider uniform subspaces in homogeneous systems, i.~e., those corresponding to pure quantum error correcting codes.

Recall that the parameters of a $((n,\,K,\,d))_D$ code  satisfy the inequality~\cite{CeCl97, Rains99}
\begin{equation}\label{Sing}
    K\leqslant D^{n-2(d-1)},
\end{equation}
which is called the \emph{quantum Singleton bound}. If parameters of a code saturate the bound in Eq.~(\ref{Sing}), the code is called \emph{quantum maximum distance separable code}~(QMDS)~\cite{Rains99}. It is known that all QMDS codes are pure~\cite{Rains99}.

In Ref.~\cite{HuGra20} an important observation about    QMDS codes was made.  We reformulate it here in a more general form and provide the proof.
\begin{lemma}[An observation in the proof of Proposition 7 of Ref.~\cite{HuGra20}]\label{TrQMDS}

    Let $((n,\,K,\,d))$ be a QMDS code.
    Consider the projector $P=\sum_{s=1}^K\,\dyad{\psi_s}$ on the codespace, where $\{\ket{\psi_s}\}$ is an orthonormal set of vectors that span the codespace. Then each reduction of $P$ to $n-(d-1)$ parties is proportional to the maximally mixed operator $I_{n-(d-1)}$.
\end{lemma}
\begin{proof}
According to Theorem~20 of Ref.~\cite{Rains98} (or Theorem~\ref{TrGen} here), tracing out one party yields a projector on a $KD$-dimensional subspace of $\ts{\C^D}{n-1}$, hence after $d-1$ such steps of tracing out we have a projector on a subspace of $\ts{\C^D}{(n-(d-1))}$ with dimension $KD^{d-1}=D^{n-(d-1)}$, i.~e. the projector on the whole space $\ts{\C^D}{(n-(d-1))}$, the identity operator.
 \end{proof}
 We stress that this  holds \emph{only} for QMDS codes - those with $K=D^{n-2(d-1)}$, and not for other pure codes. 

With the use of QMDS codes we can now formulate the following property. 
\begin{theorem}\label{codeampl}
    Let $((n_1,\, K_1,\,d_1))_{D_1}$ and $((n_2,\, K_2,\,d_2))_{D_2}$ be two QMDS codes with $K_1=K_2\equiv K>1$. Denote $r_1\equiv d_1-1$ and $r_2\equiv d_2-1$. Then there exists an $l$-uniform state in  $\ts{\C^{D_1}}{n_1}\otimes\ts{\C^{D_2}}{n_2}$ Hilbert space with 
    \begin{equation}\label{lUni}
        l=\min\left(n_1-r_1,\,n_2-r_2,\,r_1+r_2+1\right).
    \end{equation}
\end{theorem}
\begin{proof}
    Since the two codes are pure, there are two subspaces related to them:  an $r_1$-uniform subspace $W_1$ of $\ts{\C^{D_1}}{n_1}$ and  an $r_2$-uniform subspace $W_2$ of $\ts{\C^{D_2}}{n_2}$ such that $\dim(W_1)=\dim(W_2)=K$.

    Consider two isometries $V_1\colon\C^K\rightarrow\ts{\C^{D_1}}{n_1}$ and $V_2\colon\C^K\rightarrow\ts{\C^{D_2}}{n_2}$ whose ranges coincide with $W_1$ and $W_2$, respectively. Now let us take any bipartite maximally entangled state $\ket{\phi}_{AB}$ in $\C^K\otimes\C^K$ and construct the state
    \begin{equation}
        \ket{\psi}_{S_1S_2}=(V_1\otimes V_2)\ket{\phi}_{AB},
    \end{equation}
    which belongs to $\ts{\C^{D_1}}{n_1}\otimes\ts{\C^{D_2}}{n_2}$. Here $S_{1}$ and $S_2$ denote the sets of the output parties of the isometries $V_{1}$ and $V_2$, respectively. We claim that the state $\ket{\psi}_{S_1S_2}$ is $l$-uniform, with $l$ as in Eq.~(\ref{lUni}).

\begin{figure}[t]
\includegraphics[scale=0.29]{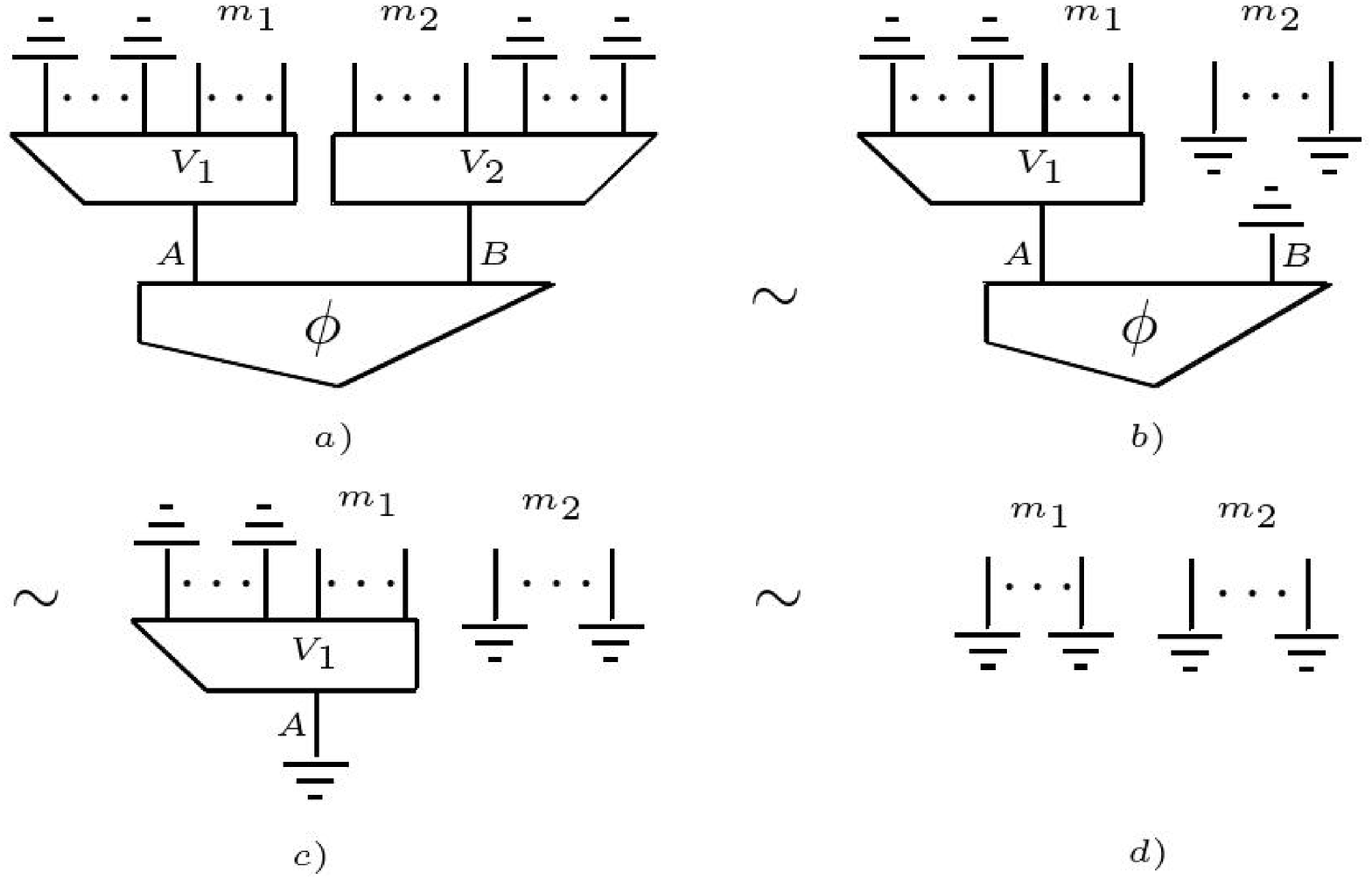}
\caption{\label{fig:lgd} Steps to prove $l$-uniformity of the state $\ket{\psi}_{S_1S_2}$.}
\end{figure}

    The underlying principle is shown on Fig~\ref{fig:lgd}. Let us assume that all subsystems of $\ket{\psi}_{S_1S_2}$ are traced out except some $m_1$ output subsystems of isometry $V_1$ and some $m_2$ output subsystems of isometry $V_2$ , as shown on Fig~\ref{fig:lgd}, a). If, for example, $m_2\leqslant r_2$, the rule from Fig.~\ref{fig:rIso} can be applied and we arrive at the situation shown on Fig~\ref{fig:lgd}, b), where isometry $V_2$ is eliminated and the state $\ket{\phi}_{AB}$ gets partially traced out. Next, the state $\ket{\phi}_{AB}$ is maximally entangled, and so its reduction  to $A$ is the maximally mixed operator $\frac1KI_A$. As a result, isometry $V_1$ acts on the identity operator $I_A$, as shown on  Fig~\ref{fig:lgd}, c). The steps $b-c)$, without taking into account the trace over output parties of $V_1$, can be written as
    \begin{equation}
        V_1\mathrm{Tr}_B\{\dyad{\phi}_{AB}\}V_1^{\dagger}=\frac1K\,V_1V_1^{\dagger}=\frac1K\,P_{\scriptscriptstyle W_1},
   \end{equation}
   where $P_{\scriptscriptstyle W_1}$ -- the orthogonal projector on subspace $W_1$, the first equality follows from maximal entanglement of $\ket{\phi}_{AB}$, the second one -- from the fact that the isometry $V_1$ has subspace $W_1$ as its range. Now if $m_1\leqslant n_1-r_1$ then, by Lemma~\ref{TrQMDS}, performing the trace over $n_1-m_1$ output subsystems of $V_1$~(Fig~\ref{fig:lgd}, c)) produces the maximally mixed state of $m_1$ parties~(Fig~\ref{fig:lgd}, d)) in addition to the maximally mixed state of $m_2$ parties obtained earlier in the first step. We conclude that if $m_1\leqslant n_1-r_1$ and $m_2\leqslant r_2$, the  reduced state of $m_1+m_2$ parties is maximally mixed. The roles of $V_1$ and $V_2$ can be interchanged, and we obtain that if $m_1\leqslant r_1$ and $m_2\leqslant n_2-r_2$, the reduced state is maximally mixed. 
   
   Now we need to determine the maximal number $l$ such that \emph{any} partition of $l=m_1+m_2$ into $m_1$ output parties of $V_1$ and $m_2$ output parties of $V_2$ yields, after performing the trace over the rest $n_1+n_2-l$ subsystems, the maximally mixed state of $l$ parties. Let us first consider partitions in which $m_2=0$. In this case the maximal value of $m_1$, for which the scheme on Fig.~\ref{fig:lgd} can still be applied,  is $n_1-r_1$, as it was shown above. This number is then an upper bound on $l$. Interchanging the roles of  $V_1$ and $V_2$ and setting $m_1=0$, we obtain another bound, $n_2-r_2$, and hence
   \begin{equation}\label{minbound}
       l\leqslant\min(n_1-r_1,\,n_2-r_2).
   \end{equation}

   Next, let us assume that $r_2>r_1$. Let $\ceil{x}$ denote the ceiling of $x$ and $\floor{x}$ denote the floor of $x$. Consider a partition of $l$  into  $m_1=\floor{\frac l2}$  and  $m_2=\ceil{\frac l2}$. If $\ceil{\frac l2}> r_2$ then $\floor{\frac l2}>r_1$ also holds, and  one cannot apply the scheme on Fig.~\ref{fig:lgd} since  neither  $V_2$ nor  $V_1$ can be eliminated in the first step a)-b) with the use of the rule from Fig.~\ref{fig:rIso}. Consequently, we can take into account only those values of $l$ that satisfy $\ceil{\frac l2}\leqslant r_2$. Accordingly, consider  a partition of $l$ into   $m_2=\ceil{\frac l2} + \alpha$ and $m_1=\floor{\frac l2}-\alpha$ for some integer $\alpha>0$ such that  $\ceil{\frac l2} + \alpha=r_2 +1$. In this case $V_2$ cannot be eliminated in  the first step of scheme on Fig.~\ref{fig:lgd}. On the other hand, the scheme can be initiated by applying the rule from Fig.~\ref{fig:rIso} with respect to $V_1$ on condition that $m_1=\floor{\frac l2}-\alpha\leqslant r_1$. If the condition is satisfied, in the step c)-d) of the scheme~(with interchanged $V_1$ and $V_2$) the reduction of $P_{\scriptscriptstyle W_2}$ to $m_2$ parties will be maximally mixed by Lemma~\ref{TrQMDS}, since $m_2\leqslant l\leqslant n_2-r_2$ by Eq.~(\ref{minbound}). This principle  continues to work for greater values of $\alpha$ (but bounded by the condition $\ceil{\frac l2} + \alpha=m_2\leqslant n_2-r_2$), as $m_1$ gets smaller. It is clear that partitions with $\alpha<0$ will also work, as the step a)-b) will be initiated with the use of $V_2$. To sum up, the maximal possible value of $l$ satisfies
   \begin{equation}\label{obtl}
       \ceil[\Big]{\frac l2} + \alpha=r_2 +1,\quad
       \floor[\Big]{\frac l2}-\alpha= r_1.
   \end{equation}
   Adding these two equalities, we obtain
   \begin{equation}\label{lfound}
       l = \ceil[\Big]{\frac l2} +  \floor[\Big]{\frac l2} = r_1 + r_2 + 1.
   \end{equation}

   When $r_1=r_2\equiv r$, we can choose~(odd) $l$ such that $\ceil{\frac l2}=r+1$ and $\floor{\frac l2}=r$. For partitions with $m_1=\floor{\frac l2}$ and $m_2 = \ceil{\frac l2}$ and, vice versa, $m_1 = \ceil{\frac l2}$ and $m_2=\floor{\frac l2}$, the scheme on Fig.~\ref{fig:lgd} is initiated with the use of $V_1$ and $V_2$, respectively. For all other partitions, which can be parameterized with integer $\alpha$ as  $m_1=\ceil{\frac l2}+\alpha$ and  $m_2=\floor{\frac l2}-\alpha$ or vice versa, the scheme also works by the analysis similar to that in the above paragraph. Consequently,
   \begin{equation}
       l = \ceil[\Big]{\frac l2} +  \floor[\Big]{\frac l2} = 2r + 1,
   \end{equation}
   which is just a special case of Eq.~(\ref{lfound}).
\end{proof}

Immediate application of Theorem~\ref{codeampl}, with the use of the correspondence between $r$-uniform states and $1$-dimensional pure quantum codes, produces
\begin{corollary}\label{codecol}
    Let $((n,\,K,\,d))_D$ be a QMDS code with $K>1$. Then there exists a pure $((2n,\,1,\,d'))_D$ code with distance
    \begin{equation}
        d'=\min(n-d+2,\,2d).
    \end{equation}
\end{corollary}

As an example, consider a $((4,\,4,\,2))_2$ code, which can be the stabilizer $[[4,\,2,\,2]]_2$ code obtained from the well-known $[[5,\,1,\,3]]_2$ code with the use of Theorem~20 of Ref.~\cite{Rains98}. Combining  $[[4,\,2,\,2]]_2$ with itself produces an $8$-qubit  $3$-uniform state. It is known that $\floor{\frac n2}$-uniform states of $n$ qubits~(\emph{absolutely maximally entangled~(AME) states}) don't exist for $n>6$~\cite{Scott04,HGS17,Rains98,Rains99a}, so the constructed state has maximal possible uniformity parameter. 

To give an example with heterogeneous systems, let us combine pure codes $((4,\, 4,\, 2))_2$ and $((5,\,4,\,3))_4$. The latter code can be produced by tensoring $((5,\,2,\,3))_2$ with itself~(Theorem 14 of Ref.~\cite{Rains99}). By the construction in the proof of Theorem~\ref{codeampl}, the two codes yield a $3$-uniform state in $\ts{\C^2}{4}\otimes\ts{\C^4}{5}$.  We can then obtain $r$-uniform subspaces by eliminating some parties of this state, but in this case some produced subspaces will demonstrate better values of $r$ than those predicted by Theorem~\ref{TrGen} owing to the following observation.
\begin{figure}[b]
\includegraphics[scale=0.37]{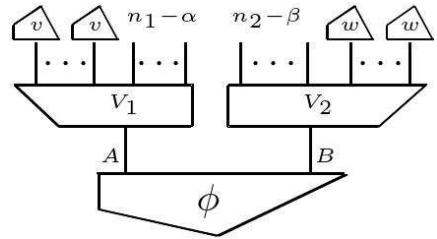}
\caption{\label{fig:ab} Construction of the states which span the $l$-uniform subspace $W$ of $\ts{\C^{D_1}}{n_1-\alpha}\otimes\ts{\C^{D_2}}{n_2-\beta}$ Hilbert space.}
\end{figure}
\begin{corollary}\label{abcor}
     Let $((n_1,\, K_1,\,d_1))_{D_1}$ and $((n_2,\, K_2,\,d_2))_{D_2}$ be two QMDS codes with $K_1=K_2\equiv K>1$. Denote $r_1\equiv d_1-1$ and $r_2\equiv d_2-1$. Then for any integers $0\leqslant\alpha\leqslant r_1$ and $0\leqslant\beta\leqslant r_2$ there exists an $l$-uniform subspace $W$ of $\ts{\C^{D_1}}{n_1-\alpha}\otimes\ts{\C^{D_2}}{n_2-\beta}$ Hilbert space such that 
     \begin{multline}\label{abparams}
         \dim(W)=D_1^{\,\,\alpha}D_2^{\beta},\\
        l=\min(n_1-r_1-\alpha,\,n_2-r_2-\beta,\\
        r_1+r_2+1-\alpha-\beta). 
     \end{multline}
\end{corollary}
\begin{proof}
    By Theorem~\ref{codeampl}, we can construct a state in  $\ts{\C^{D_1}}{n_1}\otimes\ts{\C^{D_2}}{n_2}$ Hilbert space with the uniformity parameter given by Eq.~(\ref{lUni}). Next,   
    we eliminate $\alpha+\beta$ parties of this state by the  procedure   described after the proof of Theorem~\ref{TrGen}. Let us take $\alpha$ orthonormal systems of vectors $\{\ket{v_i^{(\mu)}}\}$, $\mu=1,\,\ldots,\,\alpha$, $i=0,\,\ldots,\,D_1-1$, each system being a basis for the corresponding $\C^{D_1}$ Hilbert space. Similarly, we take $\beta$ orthonormal systems of vectors $\{\ket{w_j^{(\nu)}}\}$, $\nu=1,\,\ldots,\,\beta$, $j=0,\,\ldots,\,D_2-1$, each in its own $\C^{D_2}$ Hilbert space. Next,  we pick some specific vectors $\ket{v^{(1)}_{i_1}},\,\ldots,\,\ket{v^{(\alpha)}_{i_{\alpha}}}$, one from each system, and eliminate $\alpha$ output parties of the isometry $V_1$ by joining them with the inputs of the chosen vectors. Similarly, we pick  $\beta$ specific vectors $\ket{w^{(1)}_{j_1}},\,\ldots,\,\ket{w^{(\beta)}_{j_{\beta}}}$ and eliminate $\beta$ output parties of the isometry $V_2$~(see Fig.~\ref{fig:ab}, the indices of vectors $v,\,w$ are omitted). As a result, we obtain
    \begin{equation}\label{ns}   \bra{v^{(1)}_{i_1}}\ldots\bra{v^{(\alpha)}_{i_{\alpha}}}\bra{w^{(1)}_{j_1}}\ldots\bra{w^{(\beta)}_{j_{\beta}}}(V_1\otimes V_2)\ket{\phi}_{AB},
    \end{equation}
    one of the $D_1^{\,\,\alpha}D_2^{\beta}$ states that span the subspace $W$ of $\ts{\C^{D_1}}{n_1-\alpha}\otimes\ts{\C^{D_2}}{n_2-\beta}$ Hilbert space. All such states are hence indexed by the numbers $i_1,\,\ldots,\,i_{\alpha},\,j_1,\,\ldots,\,j_{\beta}$, which represent the correspondence  between  tuples of vectors $v,\, w$ and the basis states of $W$.
    
    The uniformity of the state in Eq.~(\ref{ns}) can be analyzed with the use of Fig.~\ref{fig:ab} and the same reasoning as in the proof of Theorem~\ref{codeampl}. The vectors $v$ and $w$ take up $\alpha$ and $\beta$ positions out of $n_1$ and $n_2$ output parties of the isometries $V_1$ and $V_2$, respectively. The parties in these positions cannot be traced out, and this results in modifying the bounds on $l$ in Eq.~(\ref{minbound}):
    \begin{equation}
       l\leqslant\min(n_1-r_1-\alpha,\,n_2-r_2-\beta).
   \end{equation}

Eq.~(\ref{obtl}) is also modified, with $r_1$ and $r_2$ replaced by $r_1-\alpha$ and $r_2-\beta$, respectively. As a result, we obtain the expression for $l$ in Eq.~(\ref{abparams}).
\end{proof}

Earlier a $3$-uniform state in $\ts{\C^2}{4}\otimes\ts{\C^4}{5}$ was obtained with the use of Theorem~\ref{codeampl} from pure codes $((4,\, 4,\, 2))_2$ and $((5,\,4,\,3))_4$. Eliminating one party with dimension $2$ and one party with dimension $4$, or, in terms of Corollary~\ref{abcor},  setting $\alpha=\beta=1$, we produce an $8$-dimensional $2$-uniform subspace of $\ts{\C^2}{3}\otimes\ts{\C^4}{4}$ Hilbert space. We stress that the original state has special structure and, as a result, after the elimination of $2$ parties the produced subspace has higher value $l=2$ in comparison with $l=1$ predicted by Theorem~\ref{TrGen}.

\subsection{\label{sec:res:comp} Comparison with mixed orthogonal arrays method and further constructions}

Mixed orthogonal arrays~(MOAs)~\cite{HPS92,HSS99} in its specific form, irredundant MOAs~(IrMOAs)~\cite{GBZ16},  have become a powerful tool in construction of $r$-uniform states in heterogeneous systems~\cite{GBZ16,PZFZ21,SSCZ22}. In this subsection we present several applications of the compositional approach that allow us to reproduce or extend some results obtained with the use of IrMOAs~(in terms of the minimal number of parties  for a given  uniformity parameter). Such a comparison also reveals some weaknesses of the presented in this paper approach.

In general it becomes more difficult to find examples of $r$-uniform states in heterogeneous systems when  the number of parties gets smaller. Let us consider some results from Ref.~\cite{PZFZ21}.
\begin{prop}[Corollary~3.2 of Ref.~\cite{PZFZ21}.]\label{prop1r}
\begin{enumerate}
$2$-uniform states exist for the following configurations:
    \item $\C^3\otimes\ts{\C^2}{N}$ for  $N\geqslant8$.
    \item 
$\ts{\C^3}{2}\otimes\ts{\C^2}{N}$ for $N\geqslant 12$.    \item
$\ts{\C^3}{3}\otimes\ts{\C^2}{N}$ for $N\geqslant 11$ and $\ts{\C^3}{4}\otimes\ts{\C^2}{N}$ for $N\geqslant 10$.
\end{enumerate}
\end{prop}

We can reproduce the first result for $N=8$. The procedure is as follows. The pure code $((5,\,2,\,3))_2$ is combined with itself by the construction of Theorem~\ref{codeampl}, which results in a $10$-qubit $3$-uniform state, i.~e. a $((10,\,1,\,4))_2$ pure code~(Corollary~\ref{codecol}). Next, by eliminating two parties, by Corollary~\ref{abcor} we obtain a $4$-dimensional $8$-qubit $2$-uniform subspace, i.~e. a pure $((8,\,4,\,3))_2$ code. Now we have an encoding isometry which maps vectors from $\C^4$ to the $2$-uniform code space~(briefly, the "$8$-qubit isometry"). Finally, we can  take a maximally entangled state in $\C^3\otimes\C^3$ and act on one of its parties with the obtained isometry, and the construction here will be similar to the one presented on Fig.~\ref{fig:23st2}. The resulting state, which belongs to $\C^3\otimes\ts{\C^2}{8}$, is $2$-uniform. For larger values of $N$ we can use the same auxiliary state and   various combinations of isometries and, if necessary, glue them together with the use of Lemma~\ref{glue}. As an example, the isometry for $N=10$, which maps $\C^4$ to  $10$-qubit $2$-uniform subspace, can be obtained from glueing the  subspace of the code $[[5,\,1,\,3]]_2$ with itself. In other words, $10$-qubit isometry is obtained from glueing $5$-qubit isometry with itself.   Next, by eliminating one party of a $((8,\,1,\,4))_2$ state, we obtain an isometry which  maps vectors from $\C^2$ to the $2$-uniform $7$-qubit space~(the "$7$-qubit isometry"). By the same procedure, from the state $((10,\,1,\,4))_2$ we obtain the $9$-qubit isometry. Now, the isometry for $N=12$ can be obtained from glueing $7$-qubit and $5$-qubit isometries, for $N=13$ -- from $5$-qubit and $8$-qubit ones, and so forth. We haven't found any appropriate isometries to construct the states with $N=9$ and $N=11$~(those that we've found have input dimension $2$, which is less than the local dimension of the second party of the auxiliary state). To conclude, we \emph{cannot} reproduce the first result of Proposition~\ref{prop1r} only for $N=9$ and $N=11$ with the current approach.

\begin{figure}[t]
\includegraphics[scale=0.3]{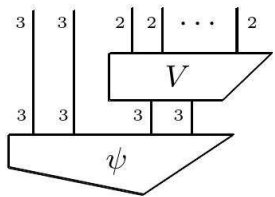}
\caption{\label{fig:332p8} An isometry $V$ acting on a joint subsystem of two parties with local dimensions $3$.}
\end{figure}

The second result from Proposition~\ref{prop1r}  is harder to reproduce.  The reason for that is as follows: we can take a $4$-qutrit $2$-uniform state, which can be, for example, the graph state of Ref.~\cite{Helw13} or a QMDS code $[[4,\,0,\,3]]_3$ from Corollary~3.6 of Ref.~\cite{JLLX10}, but now we need to act with an isometry on its two parties, i.~e., on a compound subsystem with local dimension $3\times3=9$, as shown on Fig.~\ref{fig:332p8}. The $8$-qubit isometry that was used before is not appropriate here since it has input dimension equal to $4$, which is less than the output dimension of the two combined parties. A proper isometry can be constructed from other error correcting codes with the use of the splitting property for  $r$-uniform subspaces, which is a direct consequence of the splitting method for $r$-uniform states appeared earlier in Refs.~\cite{GBZ16,SSCZ22}.
\begin{lemma}\label{splemma}
  Let $W$ be an $r$-uniform subspace of Hilbert space with the set $S=\{d_1,\,\ldots,\,d_n\}$ of local dimensions.  Let $d_i=d'_id''_i$ for some $i\colon1\leqslant i\leqslant n$ and some integer $d'_i,\,d''_i > 1$. Then there exists an $r$-uniform subspace of Hilbert space with the set of local dimensions given by $\{d'_i,\,d''_i\}\cup \left[S\setminus \{d_i\}\right]$ and having the same dimension as the original subspace. 
  \end{lemma}
  \begin{proof}
     The subspace in question can be obtained from the original one by splitting the $i$-th subsystem of each state in $W$ into two smaller ones, $i'$ and $i''$, with local dimensions $d'_i$ and $d''_i$, respectively.  Each  newly obtained state is $r$-uniform, as  follows from the splitting method described in Refs.~\cite{GBZ16,SSCZ22}. Consequently, a subspace, which consists of such states, is $r$-uniform.
  \end{proof}
  
Now we can return to the construction of a $2$-uniform state in $\ts{\C^3}{2}\otimes\ts{\C^2}{12}$. Consider a pure $((6,\,16,\,3))_4$ code~(Corollary~3.6 of Ref.~\cite{JLLX10}). By splitting each ququart into $2$ qubits, by Lemma~\ref{splemma}, the code is converted into a pure $((12,\,16,\,3))_2$ code. Since its encoding isometry~(the "$12$-qubit isometry") has input dimension equal to $16$, we can act with it  on a compound subsystem consisting of two combined parties of the state $[[4,\,0,\,3]]_3$~(see Fig.~\ref{fig:332p8}). The resulting state is $2$-uniform. The isometries for larger $N$ can be obtained from glueing the $12$-qubit isometry with the described above isometries. As an example, a $17$-qubit isometry is obtained from glueing the $12$-qubit and the $5$-qubit ones~(it doesn't matter that the input dimension of the $5$-qubit isometry is $2$ -- the input dimension of the $12$-qubit isometry is $16$, and the resulting one will have the input dimension  equal to $16\times2=32$). $N=16$ is obtained from glueing the $8$-qubit isometry with itself. All other numbers $N\geqslant 18$ can be otained similarly. In addition, $N=14$ can be constructed from splitting the code  $[[7,\,3,\,3]]_4$~(Corollary~3.6 of Ref.~\cite{JLLX10}). We \emph{cannot} reproduce the second result of Proposition~\ref{prop1r} only for $N=13$ and $N=15$.

As for the third result of Proposition~\ref{prop1r}, $2$-uniform states in $\ts{\C^3}{3}\otimes\ts{\C^2}{N}$ can be obtained with action of the described above isometries on one party of the state $[[4,\,0,\,3]]_3$. As earlier, the cases $N=9$ and $N=11$ are not covered by our approach, but we can construct a state with $N=8$, which extends the proposition. The result for uniform states in $\ts{\C^3}{4}\otimes\ts{\C^2}{N}$ can be substantially extended. Consider a code $[[4,\,0,\,3]]_6$, for example, from Corollary~3.6 of Ref.~\cite{JLLX10}. By Lemma~\ref{splemma}, by splitting each subsystem with local dimension $6$ into  qubit and qutrit subsystems, we obtain a $2$-uniform state in $\ts{\C^3}{4}\otimes\ts{\C^2}{4}$. The case $N\geqslant 5$ is trivial: we can glue the state $[[4,\,0,\,3]]_3$ with a $2$-uniform state of $N$ qubits, which exists for $N\geqslant 5$ and can be obtained, for example, from graph states constructions. These observations extend the proposition from $N=10$ to $N=4$.

Gathering the above results, we can formulate
\begin{prop}[Combination of Corollary~3.2 of  Ref.~\cite{PZFZ21} with the current approach]
$2$-uniform states exist for the following configurations: 
\begin{enumerate}
    \item $\C^3\otimes\ts{\C^2}{N}$ for  $N\geqslant8$.
    \item 
$\ts{\C^3}{2}\otimes\ts{\C^2}{N}$ for $N\geqslant 12$.    \item
$\ts{\C^3}{3}\otimes\ts{\C^2}{N}$ for $N=8$ and $N\geqslant 10$ and $\ts{\C^3}{4}\otimes\ts{\C^2}{N}$ for  $N\geqslant 4$.
\end{enumerate}
\end{prop}

Let us also analyze some results of Ref.~\cite{SSCZ22}.
\begin{prop}[Theorem 9 of Ref.~\cite{SSCZ22}]\label{prop23}
For any $d>2$, the following holds.
\begin{enumerate}
    \item There exists a $2$-uniform state in $\ts{\C^2}{2}\otimes\ts{\C^d}{N}$ for any $N\geqslant7$ and $N\ne 4d+2,\,4d+3$.
    \item There exists a $2$-uniform state in $\C^2\otimes\ts{\C^d}{N}$ for any $N\geqslant5$.
\end{enumerate}
\end{prop}
We can start with the $2$-uniform subspace of the code $[[5,\,1,\,3]]_2$ and act with a proper  isometry on three subsystems, i.~e. on a joint system of dimension $8$, of each vector in the code. Therefore, in addition to having a $2$-uniform subspace as its range, an appropriate isometry must have input dimension greater or equal $8$. The code family $[[N,\,N-4,\,3]]_d$, $4\leqslant N\leqslant d^2+1$, $d>2$~(Corollary~3.6 of Ref.~\cite{JLLX10}) provides us with proper isometries for $6\leqslant N\leqslant10$. In addition, the isometry with $5$ output parties, which corresponds to $[[5,\,1,\,3]]_d$, only works when  $ d\geqslant 8$, since in this case its input dimension is equal to $d$. Isometries for all other numbers, $N>10$, can be obtained by glueing the codes with $N<10$~(Lemma~\ref{glue}). As a result, we lift the constraint $N\ne 4d+2,\,4d+3$ and obtain $2$-uniform subspaces instead of just states.  

We can only reproduce the second result of Proposition~\ref{prop23}. All the described above isometries, this time including the one with $N=5$, can be used to act on one party of a maximally entangled state in $\C^2\otimes\C^2$.

Instead of a maximally entangled state in $\C^2\otimes\C^2$ we could use maximally entangled subspaces of $\C^2\otimes\C^p$, which have dimension equal to $\floor{\frac p2}$, by Corollary 4 of Ref. [7]. Now the isometry, which corresponds to code $[[N,\,N-4,\,3]]_d$, acts on a party with local dimension $p$ of each state in the maximally entangled subspace. The input dimension of the isometry, $d^{N-4}$, hence must be greater or equal $p$, and we have the condition
\begin{equation}
    N\geqslant 4+\log_d p.
\end{equation}

Summing the results, we can formulate the extension of Proposition~\ref{prop23} 
\begin{prop}
The following holds.
\begin{enumerate}
    \item For $2<d\leqslant 8$ there exists a $2$-uniform subspace of $\ts{\C^2}{2}\otimes\ts{\C^d}{N}$ with dimension $2$ for any $N\geqslant 6$.
    \item For $d>8$  there exists a $2$-uniform subspace of $\ts{\C^2}{2}\otimes\ts{\C^d}{N}$ with dimension $2$ for any $N\geqslant 5$.
    \item for $d>2$ and $p\geqslant2$ there exists a $2$-uniform subspace of $\C^2\otimes\ts{\C^d}{N}$ with dimension $\floor{\frac p2}$   for any $N\geqslant 4+\log_d p$.
\end{enumerate}
\end{prop}

The above examples show that the presented approach is more effective in constructing $r$-uniform states and subspaces with larger local dimensions, i.~e., qutrits or higher. Indeed, there are not many qubit isometries for a given value of the uniformity parameter, and, in reproducing some results of Proposition~\ref{prop1r}, we had to resort to splitting the codes of higher dimensionality. A similar tendency was observed in Ref.~\cite{KVAnt21} where genuinely entangled subspaces were constructed with the isometric mapping method: when local dimension goes to infinity, the dimension of the obtained subspaces asymptotically approaches the maximal possible value.

Finally, let us provide some constructions with higher values of the uniformity parameter.

Consider the pure QMDS code $[[10,\,4,\,4]]_3$ from Theorem 13 of Ref.~\cite{GR15}. From Corollary~\ref{abcor} with $\alpha=\beta=1$, we obtain a $5$-uniform $9$-dimensional subspace of $\ts{\C^3}{18}$. The corresponding isometry $V$ has the input dimension equal to $9$. Let us take a maximally entangled subspace of $\C^2\otimes\C^9$, which has dimension equal to $\floor{\frac92}=4$~(Corollary 4 of Ref. [7]). Action of $V$ on the second party of each state in this subspace yields a $5$-uniform $4$-dimensional subspace of $\C^2\otimes\ts{\C^3}{18}$. With the same isometry $V$ we could act instead on the joint subsystem of three parties of each state in the code space $[[5,\,1,\,3]]_2$, and this procedure yields a  $5$-uniform $2$-dimensional subspace of $\ts{\C^2}{2}\otimes\ts{\C^3}{18}$.

Consider the $[[10,\,0,\,6]]_4$ code, which can be obtained from the classical  $[10,\,5,\,6]$  MDS code over $GF(16)$ of Ref.~\cite{GKL08} by the correspondence between stabilizer QMDS codes and self-dual classical MDS codes~(Theorem 15 of Ref.~\cite{KKKS06}, see also Proposition~15 of Ref.~\cite{HuGra20}). By elimination of one party a code $((9,\,4,\,5))_4$ is constructed~(Theorem~20 of Ref.~\cite{Rains98}). By splitting the latter code we obtain a $((18,\,4,\,5))_2$ code whose encoding isometry has the input dimension equal to $4$. Applying this isometry to one of the parties of a maximally entangled state in $\C^3\otimes\C^3$ produces a $4$-uniform state in $\C^3\otimes\ts{\C^2}{18}$.

\section{\label{sec:conc} Discussion}

In this paper we've shown how new $r$-uniform states and subspaces can be constructed from combining already known quantum error correcting codes, (maximally)~entangled states and subspaces. The isometric mapping method played the key role here: one takes an isometry, which, as its range, has a subspace with some useful property, and applies it to states or  subspaces, perhaps with some other interesting property. This approach allowed us to complement some results which were obtained with the mixed orthogonal arrays method. It would be interesting to continue this parallel with OAs. One example in this direction could be analyzing  encoding isometries of the QECCs obtained with OAs, for instance, the ones from Ref.~\cite{PXC22}. This could potentially lead to new OA and MOA constructions.

The advantage of the presented approach is its experimental accessibility:  whenever one can realize  encoding isometries of QECCs as well as prepare auxiliary entangled states, one can construct uniform states in accordance with the described above procedures. The disadvantage of the approach is that it doesn't utilize the internal structure of the combined objects  beyond their uniformity property. Taking more structural properties into account could result in constructing more classes of useful states and subspaces such as, for example,  AME states, the ones we couldn't produce with the current approach. This observation suggests another direction of further research: how to combine several QECCs in the most efficient way, with taking their specific properties into account, to obtain a new QECC with ``good'' characteristics~(in the sense similar to the recent ``good quantum codes'' constructions~\cite{HHD20,PK21}). We stress that the distance of the codes composed by the procedure of Theorem~\ref{codeampl} doesn't scale with the number of codes being combined: the distance of the resulting code will always be upper-bounded by the minimum of the number in Eq.~(\ref{minbound}) taken over all the codes being combined.  

It also would be interesting to apply the isometric mapping method to construction of multipartite subspaces  with another useful property -- distillability and closely related to it non-positivity of partial transpose across each bipartition~(distillable and NPT subspaces). This direction of research could complement the results obtained in Refs.~\cite{NJohn13,AgHalBa19,JLP19,MKA22}.

\section*{Acknowledgments}
The author thanks M. V. Lomonosov Moscow State University for supporting this work.

\nocite{*}

\bibliography{refs}

\end{document}